\documentclass[a4paper, 10pt]{article}
\usepackage{a4wide}
\usepackage{mathptmx}
\usepackage{authblk}

\def\hb{\hbox to 10.7 cm{}}

\usepackage{amsmath, amssymb, amsthm, proof}
\usepackage{thmtools,thm-restate}
\usepackage{paralist}
\usepackage{subcaption}
\usepackage{tikz}
\usetikzlibrary{arrows.meta}

\usepackage{paralist,enumitem}
\setlist[itemize]{leftmargin=*,noitemsep}

\newcommand{\Ra}{\Rightarrow} 
\newcommand{\impl}{\supset} 
\newcommand{\nc}{{\mathrel|\joinrel\sim}} 
\newcommand{\nnc}{{\mathrel|\joinrel\not\sim}} 
\newcommand{\tuple}[1]{\left\langle #1 \right\rangle}
\newcommand{\ncass}[1]{\nc_{{\sf A},{#1}}}
\newcommand{\nncass}[1]{\nnc_{{\sf A},{#1}}}
\newcommand{\ncaba}[1]{\nc_{{\sf ABA},{#1}}}
\newcommand{\nncaba}[1]{\nnc_{{\sf ABA},{#1}}}
\newcommand{\ncmcs}{\nc_{\mathsf{mcs}}}
\newcommand{\nncmcs}{\nnc_{\mathsf{mcs}}}

\newcommand{\AFLARS}{{\sf AF}_{{\sf L},{\sf AR}}({\sf S})} 
\newcommand{\AFLS}{{\sf AF}_{\sf L} ({\sf S})} 
\newcommand{\AFLSAS}{{\sf AF}_{\sf L} ({\sf S},{\sf A})} 
\newcommand{\AFABA}{{\sf AF}_{\left\langle\mathcal{L},\mathcal{R}\right\rangle}({\sf S},{\sf A})} 
\newcommand{\AFseqABA}{\calAF_{\tuple{\calL,\calR_\Rightarrow}}^{\sf ABA_\Ra}(\calS,\calA)} 

\newcommand{\ArgL}{\text{Arg}_{\sf L}} 
\newcommand{\ArgLS}{\text{Arg}_{\sf L}({\sf S})} 
\newcommand{\ArgLSAS}{\text{Arg}_{{\sf L}}({\sf S},{\sf A})} 
\newcommand{\ArgLRSA}{\text{Arg}_{\left\langle\mathcal{L},\mathcal{R}\right\rangle}^{\sf ABA}({\sf S},{\sf A})} 
\newcommand{\ArgLR}{\text{Arg}_{\left\langle\mathcal{L},\mathcal{R}\right\rangle}^{\sf ABA}} 
\newcommand{\ArgLRseqSA}{\text{Arg}_{\left\langle\mathcal{L},\mathcal{R}_\Rightarrow\right\rangle}^{\sf ABA_{\Rightarrow}}({\sf S},{\sf A})} 
\newcommand{\ArgLRseq}{\text{Arg}_{\left\langle\mathcal{L},\mathcal{R}_\Rightarrow\right\rangle}^{\sf ABA_{\Rightarrow}}} 

\newcommand{\ARule}{{\sf AR}} 
\newcommand{\ATtack}{{\sf AT}} 
\newcommand{\ext}{\mathcal{E}} 
\newcommand{\sem}{{\sf sem}} 
\newcommand{\grd}{{\sf grd}} 
\newcommand{\cmp}{{\sf cmp}} 
\newcommand{\prf}{{\sf prf}} 
\newcommand{\stb}{{\sf stb}} 

\newcommand{\conc}{{\sf Conc}} 
\newcommand{\concs}{{\sf Concs}} 
\newcommand{\supp}{{\sf Supp}} 
\newcommand{\supps}{{\sf Supps}} 
\newcommand{\sub}{{\sf Sub}} 
\newcommand{\CN}{{\sf CN}} 
\newcommand{\free}{{\sf Free}} 
\newcommand{\exts}{{\sf Ext}} 
\newcommand{\ass}{{\sf Ass}} 
\newcommand{\MCS}{{\sf MCS}} 

\newcommand{\logic}{$\textsf{L} = \left\langle\mathcal{L},\vdash\right\rangle$}
\newcommand{\calc}{$\mathsf{C}$} 
\newcommand{\CL}{{\sf CL}} 
\newcommand{\LK}{{\sf LK}} 

\newcommand{\calA}{{\sf A}} 
\newcommand{\calC}{\mathcal{C}} 
\newcommand{\calL}{\mathcal{L}} 
\newcommand{\calR}{\mathcal{R}} 
\newcommand{\calS}{{\sf S}} 
\newcommand{\calT}{{\sf T}} 
\newcommand{\calAF}{{\sf AF}} 
\newcommand{\AS}{{\sf A}} 
\newcommand{\sfL}{{\sf L}} 
\newcommand{\sfS}{{\cal S}} 

\DeclareFontFamily{U}{matha}{}
\DeclareFontShape{U}{matha}{m}{n}{
  <-5.5>    matha5
  <5.5-6.5> matha6 
  <6.5-7.5> matha7
  <7.5-8.5> matha8
  <8.5-9.5> matha9
  <9.5-11>  matha10
  <11->     matha12
}{}
\DeclareSymbolFont{matha}{U}{matha}{m}{n}
\DeclareFontSubstitution{U}{matha}{m}{n}
\DeclareFontFamily{U}{mathx}{\hyphenchar\font45}
\DeclareFontShape{U}{mathx}{m}{n}{<-> mathx10}{}
\DeclareSymbolFont{mathx}{U}{mathx}{m}{n}
\DeclareFontSubstitution{U}{mathx}{m}{n}

\DeclareMathDelimiter{\lfilet}     {4}{mathx}{"37}{mathx}{"37}

\newcommand{\acom}{\mathrel\big\lfilet}

\newtheorem{definition}{Definition}
\newtheorem{example}{Example}
\newtheorem{notation}{Notation}
\newtheorem{proposition}{Proposition}

\newtheorem{lemma}{Lemma}
\newtheorem{remark}{Remark}

\usepackage[normalem]{ulem}

\begin{document}



\title{Equipping sequent-based argumentation with \\defeasible assumptions}


\author{AnneMarie Borg\thanks{The author is supported by the Alexander von Humboldt Foundation and the German Ministry for Education and Research, and the Israel Science Foundation (grant 817/15).}}

%
\affil{Institute for Philosophy II, Ruhr-University Bochum, Germany}
\date{}

\maketitle

\begin{abstract}
In many expert and everyday reasoning contexts it is very useful to reason on the basis of defeasible assumptions. For instance, if the information at hand is incomplete we often use plausible assumptions, or if the information is conflicting we interpret it as consistent as possible. In this paper sequent-based argumentation, a form of logical argumentation in which arguments are represented by a sequent, is extended to incorporate assumptions. The resulting assumptive framework is general, in that some other approaches to reasoning with assumptions can adequately be represented in it. To exemplify this, we show that assumption-based argumentation can be expressed in assumptive sequent-based argumentation. 
\end{abstract}


\section{Introduction}

Assumptions are an important concept in defeasible reasoning. Often, in both expert and everyday reasoning, the information provided is not complete or it is inconsistent. By assuming additional information or considering consistent subsets of information, a conclusion can be reached in such cases. 
A well-known formal method for modeling defeasible reasoning is abstract argumentation theory, introduced by Dung~\cite{Dung95}. 
In logical argumentation, the arguments have a specific structure on which the attacks depend~\cite{BesHun01,Pra10}. One such logical argumentation framework is \emph{sequent-based argumentation}~\cite{ArStr15argcomp}, in which arguments are represented by sequents, as introduced by Gentzen~\cite{Gen34} and well-known in proof theory. Attacks between arguments are formulated by \emph{sequent elimination rules}, which are special inference rules. The resulting framework is generic and modular, in that any logic, with a corresponding sound and complete sequent calculus can be taken as the deductive base (the so-called \emph{core logic}). 


In this paper we extend sequent-based argumentation. To each sequent a component for assumptions is added. This way, a distinction can be made between strict and defeasible premises, to reach further conclusions. As an instance of the obtained framework, assumption-based argumentation (ABA)~\cite{BDKT97,DKT09,Toni14} is studied and the relation to reasoning with maximally consistent subsets~\cite{ReMa70} is investigated. The latter is a well-known method to maintain consistency, in view of inconsistent information. 
ABA is a structural argumentation framework which is also abstract, in that there are only limited assumptions on the underlying deductive system. It was introduced to determine a set of assumptions that can be accepted as a conclusion from the given information. 

Arguments in ABA are constructed by applying modus ponens to simple clauses of an inferential database. Only recently logic-based instantiations of ABA have been studied, mostly with classical logic as the core logic. Sequent-based argumentation, and the here introduced assumptive generalization, are more general and modular, in that these are based on a Tarskian core logic and the arguments are constructed via the inference rules of the corresponding sequent calculus. Logics that can be equipped with defeasible assumptions by means of assumptive sequent-based argumentation include, in addition to classical logic, intuitionistic logic, many of the well-known modal logics and several relevance logics. Hence, the results of this paper generalize to many deductive core systems, as long as the Tarskian conditions are fulfilled. 

Sequent calculi and sequent-based argumentation have some further advantages as well. For example, the latter comes equipped with a dynamic proof theory~\cite{ArStr15LFSA,ArStr17DD}, introduced to study argumentation from a proof theoretical perspective. These dynamic derivations provide a mechanism for deriving arguments as well as attacks and hence to reach conclusions for a given argumentation framework in an automatic way. Sequent calculi themselves have been investigated for many logics and purposes, mainly in the context of proof theory. A significant advantage over other proof systems is, that the premises can be manipulated within a proof, see also~\cite{SchHei03}. 





The paper is organized as follows. In the next section, sequent-based argumentation is recalled. Then, in Section~\ref{sec:AssumptiveSeq}, the general framework for assumptive sequent-based argumentation is introduced. This framework will be considered in Section~\ref{sec:ABA}, in which ABA is taken as an example, to show how the assumptive sequent-based framework can be applied. We conclude in Section~\ref{sec:Conclusion}.

\section{Sequent-based argumentation}

Throughout the paper only propositional languages are considered, denoted by $\calL$. Atomic formulas are denoted 
by $p,q$, formulas are denoted by $\phi,\psi$, sets of formulas are denoted by $\calS,\calT$, and finite sets of 
formulas are denoted by $\Gamma,\Delta$, later on we will denote sets of assumptions by $\calA$ and finite sets of assumptions by $A$, all of which can be primed or indexed.

\begin{definition}
\label{def:logic}
A \emph{logic\/} for a language $\calL$ is a pair \logic, where $\vdash$ is a (Tarskian) consequence relation for $\calL$, having the following properties: \emph{reflexivity:} if $\phi\in\calS$, then $\calS\vdash\phi$; \emph{transitivity:} if $\calS\vdash\phi$ and $\calS',\phi\vdash\psi$, then $\calS,\calS'\vdash\psi$; and \emph{monotonicity:} if $\calS'\vdash\phi$ and  $\calS'\subseteq\calS$, then $\calS\vdash\phi$. 
\end{definition}

As usual in logical argumentation (see, e.g.,~\cite{BesHun01,Pol92,Prak96,SiLou92}), arguments have a specific structure based on the underlying formal language, the \emph{core logic\/}. In the current setting arguments are represented by the well-known proof theoretical notion of a \emph{sequent\/}. 

\begin{definition}
  \label{def:seq}
  \label{def:seqargu}
  \label{not:suppconc}
Let \logic\ be a logic and $\calS$ a set of $\mathcal{L}$-formulas.
\begin{itemize}
    \item An \emph{$\calL$-sequent} (\emph{sequent} for short) is an expression of the form $\Gamma\Ra\Delta$, where $\Gamma$ and $\Delta$ are finite sets of formulas in $\calL$ and $\Ra$ is a symbol that does not appear in $\calL$.
    \item An \emph{$\sfL$-argument} (\emph{argument} for short) is an $\calL$-sequent $\Gamma\Ra\psi$,\footnote{Set signs in arguments are omitted.} where $\Gamma\vdash\psi$. $\Gamma$ is called the \emph{support set\/} of the argument and $\psi$ its \emph{conclusion\/}.
    \item An \emph{$\sfL$-argument based on $\calS$} is an $\sfL$-argument $\Gamma\Ra\psi$, where $\Gamma\subseteq\calS$. We denote by $\ArgLS$ the set of all the $\sfL$-arguments based on $\calS$. 
\end{itemize}
Given an argument $a = \Gamma \Ra \psi$, we denote $\supp(a) = \Gamma$ and $\conc(a) = \psi$. We say that 
$a'$ is a {\em sub-argument\/} of $a$ iff  $\supp(a') \subseteq\supp(a)$. The set of all the sub-arguments of $a$ 
is denoted by $\sub(a)$.
\end{definition}

The formal systems used for the construction of sequents (and so of arguments) for a logic \logic, are \emph{sequent calculi\/}~\cite{Gen34}, denoted here by \calc. In what follows we shall assume that \calc\ is sound and complete for \logic,
i.e., $\Gamma\Ra\psi$ is provable in \calc\ iff $\Gamma\vdash\psi$. One of the advantages of sequent-based argumentation is that any logic with a corresponding sound and complete sequent calculus can be used as the core logic.\footnote{See~\cite{ArStr15argcomp} for further advantages of this approach.} The construction of arguments from 
simpler arguments is done by the \emph{inference rules\/} of the sequent calculus~\cite{Gen34}. 

Argumentation systems contain also attacks between arguments. In our case, attacks are represented by \emph{sequent elimination rules}. Such a rule consists of an attacking argument (the first condition of the rule), an attacked argument (the last condition of the rule), conditions for the attack (the conditions in between) and a conclusion (the eliminated attacked sequent). The outcome of an application of such a rule is that the attacked sequent is `eliminated'. The elimination of a sequent $a=\Gamma\Ra\Delta$ is denoted by $\overline{a}$ or $\Gamma\not\Ra\Delta$.

\begin{definition}
  \label{def:elim}
  A \emph{sequent elimination rule\/} (or \emph{attack rule}) is a rule $\calR$ of the form:
  \begin{small}
\begin{equation}
  \label{eq:elim}
    \infer[\ \ \calR]{\Gamma_n\not\Ra\Delta_n}{\Gamma_1\Ra\Delta_1 & \ldots & \Gamma_n\Ra\Delta_n}
\end{equation}
\end{small}
It is said that $\Gamma_1\Ra\Delta_1$ \emph{$\calR$-attacks} $\Gamma_n\Ra\Delta_n$.
\end{definition}

\begin{example}
  \label{ex:seqelimrule}
  Suppose $\calL$ contains a $\vdash$-negation $\neg$ (where $p\nvdash p$ and $\neg p\nvdash\neg p$ for every atom $p$) and a $\vdash$-conjunction $\wedge$ (where $\calS\vdash\phi\wedge\psi$ iff $\calS\vdash\phi$ and $\calS\vdash\psi$). We refer to~\cite{ArStr15argcomp,StrAr15logcom} for a definition of a variety of attack rules. Assuming that $\Gamma_2\neq\emptyset$, two such rules are:
  \begin{small}
\begin{align*}
  & \text{Undercut (Ucut):} 
    \quad
  \infer[]{\Gamma_2,\Gamma'_2\not\Ra\psi_2}{
    \Gamma_1\Ra\psi_1 & \Ra\psi_1\leftrightarrow \neg\bigwedge\Gamma_2 & \Gamma_2,\Gamma'_2\Ra\psi_2
  }\\
  & \text{Direct Ucut (DUcut):}
    \quad
  \infer[]{\gamma,\Gamma'_2\not\Ra\psi_2}{
    \Gamma_1\Ra\psi_1 & \Ra\psi_1\leftrightarrow \neg\gamma & \gamma,\Gamma'_2\not\Ra\psi_2
  }
\end{align*}
\end{small}
\end{example}

A sequent-based framework is now defined as follows:

\begin{definition}
  \label{def:seqAF} 
A \emph{se\-quent\--based argumentation framework} for a set of formulas $\calS$ based on the logic \logic\ and a set $\ARule$ of sequent elimination rules, is a pair $\AFLARS = \tuple{\ArgLS, \ATtack}$, where $\ATtack\subseteq\ArgLS\times\ArgLS$ and $(a_1,a_2)\in\ATtack$ iff there is an $\calR\in\ARule$ such that $a_1$ $\calR$-attacks $a_2$. 
\end{definition}

In what follows, to simplify notation, we will omit the subscripts $\sfL$ and/or $\ARule$ when these are clear from the context or arbitrary. 

\begin{example}
  \label{ex:seqargumentation}
  Let $\calAF_{\CL,\{{\sf Ucut}\}}(\calS)$ be an argumentation framework, with classical logic as its core logic, Ucut the only attack rule and the set $\calS = \{p,p\impl q,\neg q\}$. Some of the arguments are: $a = p,p\impl q\Ra q$, $\ b= \neg q\Ra\neg q$, $\ c = p\Ra p$, $\ d =\ \ \Ra q\vee\neg q$ and $e = p\impl q, \neg q\Ra\neg p$.
  
  Note that $a$ and $e$ attack each other. Morever, $a$ attacks $b$ and $e$ attacks $c$. Since $\supp(d) = \emptyset$, it follows that $d$ is not attacked at all. 
  
%
%
\end{example}

Given a (sequent-based) framework, Dung-style semantics~\cite{Dung95} can be applied to it, to determine what combinations of arguments (called \emph{extensions}) can collectively be accepted from it. 

\begin{definition} 
\label{def:extension}
Let $\AFLS =\tuple{\ArgLS,\ATtack}$ be an argumentation framework and $\sfS\subseteq\ArgLS$ a set of arguments.  $\sfS$ \emph{attacks\/} an argument $a$ if there is an $a'\in\sfS$ such that $(a',a)\in\ATtack$; $\sfS$ \emph{defends\/} an argument $a$ if $\sfS$ attacks every attacker of $a$; $\sfS$ is \emph{conflict-free\/} if there are no arguments $a_1,a_2\in\sfS$ such that $(a_1,a_2)\in\ATtack$; $\sfS$ is \emph{admissible\/} if it is conflict-free and it defends all of its elements. An admissible set that contains all the arguments that it defends is a \emph{complete extension\/} of $\AFLS$. 

Some particular complete extensions of $\AFLS$ are: a \emph{preferred extension\/} of $\AFLS$ is a maximal (with respect to $\subseteq$) complete extension of $\ArgLS$; a \emph{stable extension\/} of $\ArgLS$ is a complete extension that attacks every argument not in it; the \emph{grounded extension\/} of $\AFLS$ is the minimal (with respect to $\subseteq$) complete extension of $\ArgLS$.
\end{definition}

We denote by $\exts_{\sem}(\AFLS)$ the 
set of all the extensions of $\AFLS$ under the semantics $\sem\in\{\cmp,\grd,\prf,\stb\}$. The subscript is 
omitted when this is clear from the context or arbitrary. 

\begin{definition}
\label{def:entailment}
Given a sequent-based argumentation framework $\AFLS$, the semantics as defined in 
Definition~\ref{def:extension} induces corresponding (nonmonotonic) \emph{entailment relations}:\footnote{Since the grounded extension is unique, $\nc^\cap_{{\sf L},\grd}$, $\nc^\cup_{{\sf L},\grd}$ and $\nc^\Cap_{\sfL,\grd}$, are the same, and will be denoted by $\nc_{{\sf L},\grd}$.} 
\begin{itemize}
   \item $\calS\:\nc^\cap_{{\sf L},\sem}\:\phi$ ($\calS\:\nc^\cup_{{\sf L},\sem}\:\phi$) iff for every (some) extension 
           $\ext\in\exts_{\sem}(\AFLS)$, there is an argument $\Gamma\Ra\phi\in\ext$ for $\Gamma\subseteq\calS$,
   \item $\calS\:\nc^\Cap_{{\sf L},\sem}\:\phi$ iff for every $\ext\in\exts_\sem(\AFLS)$ there is an $a\in\ext$ and 
           $\conc(a) = \phi$.  
\end{itemize}
\end{definition}

\begin{example}
  \label{ex:seqargumentationentailment}
  Consider the framework from Example~\ref{ex:seqargumentation}, for $\calS = \{p, p\impl q, \neg q\}$ and undercut as the only attack rule.  The argument $d = \ \Ra q\vee\neg q$ is not attacked and hence $\calS\:\nc_{\CL,\grd}\: q\vee\neg q$. For the other formulas in $\phi\in\calS$ we have that $\calS\:\nnc_{\CL,\sem}^\cap\:\phi$ and $\calS\:\nc_{\CL,\sem}^\cup\:\phi$, for $\sem\in\{\cmp,\prf,\stb\}$. 
\end{example}

\subsection{Reasoning with maximally consistent subsets}
\label{sec:MCS}

Reasoning with maximally consistent subsets is a well-known way to maintain consistency when provided with an inconsistent set of formulas~\cite{ReMa70}. First some useful notions:
\begin{definition}
  \label{def:consistency}
  Let \logic\ be a logic, with at least the connectives $\neg$ and $\wedge$ (see Example~\ref{ex:seqelimrule}) and let $\calT$ be a set of $\calL$-formulas. 
  \begin{itemize}
    \item The \emph{closure} of $\calT$ is denoted by $\CN(\calT)$ (thus, $\CN(\calT) = \{\phi\mid\calT\vdash\phi\}$).
    \item $\calT$ is \emph{consistent} (for $\vdash$), if there are no formulas $\phi_1,\ldots, \phi_n\in\calT$ such that $\vdash\neg\bigwedge_{i=1}^n\phi_i$.
    \item A subset $\calC$ of $\calT$ is a \emph{minimal conflict} of $\calT$ (w.r.t $\vdash$), if $\calC$ is inconsistent and for any $c\in\calC$, $\calC\setminus\{c\}$ is consistent. $\free(\calT)$ denotes the set of formulas in $\calT$ that are not part of any minimal conflict of $\calT$. 
  \end{itemize}
\end{definition}

Denote by $\MCS_\sfL(\calS)$ the set of all maximally consistent subsets of $\calS$ for the logic $\sfL$. The subscript is omitted when arbitrary or clear from the context. 

\begin{definition}
  \label{def:MCSentailment}
  Let \logic\ and $\calS$ a set of $\calL$-formulas. Several entailment relations are then defined as follows:
  \begin{itemize}
    \item $\calS\:\ncmcs^\cap\:\phi$ iff $\phi\in\CN(\bigcap\MCS(\calS)$;
    \item $\calS\:\ncmcs^\cup\:\phi$ iff $\phi\in\bigcup_{\calT\in\MCS(\calS)}\CN(\calT)$;
    \item $\calS\:\ncmcs^\Cap\:\phi$ iff $\phi\in\bigcap_{\calT\in\MCS(\calS)}\CN(\calT)$.
  \end{itemize}
\end{definition}

\begin{example}
  \label{ex:MCSflat}
  Consider the set $\calS = \{p, p\impl q, \neg q\}$ and core logic \CL. Then there are three maximally consistent subsets: $\MCS(\calS) = \{\{p, p\impl q\}, \{p, \neg q\}, \{p\impl q, \neg q\}\}$. Hence $\bigcap\MCS(\calS) = \emptyset$. Moreover, $\calS\:\ncmcs^\cap\:\phi$ and $\calS\:\ncmcs^\Cap\:\phi$ if and only if $\phi$ is a \CL-tautology. But $\calS\:\ncmcs^\cup\:\psi$, for $\psi\in\CN(\calS)$. 
\end{example}

Recently it was shown that sequent-based argumentation is a useful platform to incorporate reasoning with maximally consistent subsets~\cite{ABS17AI,ArStr16KR}. It was shown, for $\AFLS = \tuple{\ArgLS,\ATtack}$, classical logic as core logic, undercut as attack rule and $\calS$ a set of formulas that $\calS\:\nc_\prf^\pi\:\phi$ iff $\calS\:\nc_\stb^\pi\:\phi$ iff $\calS\:\ncmcs^\pi\:\phi$, where $\pi\in\{\cap,\cup,\Cap\}$. 
Indeed, the results from Examples~\ref{ex:seqargumentationentailment} and~\ref{ex:MCSflat} are the same. 

\section{Assumptive sequent-based argumentation}
\label{sec:AssumptiveSeq}

Sometimes deriving conclusions requires making assumptions, for example, because there is simply not enough information given, or the information provided is conflicting. There are many ways in which assumptions are handled in the literature, e.g., default logic~\cite{Rei80}, assumption-based argumentation~\cite{BDKT97}, default assumptions~\cite{Mak03} and adaptive logics~\cite{Bat07}. In this section we extend the sequent-based argumentation framework from the previous section, to incorporate assumptions. This generalization is formulated in a general way: independent of the core logic, the nature of the assumptions, or the way that the system allows for deriving conclusions based on these assumptions. 

In what follows we assume that, instead of one set of formulas, the input contains two sets of $\calL$-formulas: $\AS$, a set of, possibly conflicting, assumptions or defeasible premises, the form of which depends on the application and the logic; and $\calS$, a consistent set, the formulas of which can intuitively be understood as facts or strict premises. As before, we assume that a logic \logic\ has a corresponding sequent calculus \calc. This calculus will, depending on the application, be extended to \calc$'$, in order to allow for assumptions.  

\begin{definition}
  Let \logic\ be a logic, with a corresponding sound and complete sequent calculus \calc\ and sequent calculus extension \calc$'$, let $\calS$ be a consistent set of $\calL$-formulas and $\AS$ a set of assumptions.
  \begin{itemize}
    \item An \emph{assumptive $\calL$-sequent} (\emph{(assumptive) sequent} for short) is a sequent $A\acom \Gamma\Ra\Delta$.
    \item An \emph{assumptive $\sfL$-argument} (\emph{(assumptive) argument} for short) is an assumptive sequent $A\acom\Gamma\Ra\Delta$, that is provable in \calc$'$.\footnote{Often, \calc$'$ will be the result of adding rules, to divide the support set of each argument into the set of defeasible premises on the left-hand-side and the set of strict premises on the right-hand-side of $\acom$, to \calc.} 
    \item An \emph{assumptive $\sfL$-argument based on $\calS$ and $\AS$} is an assumptive argument $A\acom\Gamma\Ra\Delta$ such that $\Gamma\subseteq\calS$ and $A\subseteq\AS$. As before, we denote by $\ArgLSAS$ the set of all the assumptive $\sfL$-arguments based on $\calS$ and $\AS$.
  \end{itemize}
  \label{def:assargument}
\end{definition}


\begin{notation}
  Let $a = A\acom\Gamma\Ra\Delta$ be an assumptive argument. Then $\ass(a) = A$ denotes the assumptions of the argument $a$. As before, $\supp(a) = \Gamma$ and $\conc(a) = \Delta$. Furthermore, for $\sfS$ a set of arguments, $\concs(\sfS) = \{\conc(a)\mid a\in\sfS\}$, $\supps(\sfS) = \bigcup\{\supp(a)\mid a\in\sfS\}$ and $\ass(\sfS) = \bigcup\{\ass(a)\mid a\in\sfS\}$. In case that $A = \emptyset$, $a$ will sometimes be written as $\Gamma\Ra\Delta$.
  \label{not:asssupcon}
\end{notation}

An important rule in sequent calculi is $[\text{Cut\/}]$. In assumptive notation there are two:
\begin{gather*}
  \infer[{[\text{Cut\/}]}]{A_1,A_2\acom\Gamma_1,\Gamma_2\Ra\Delta_1,\Delta_2}{
    A_1\acom\Gamma_1\Ra\Delta_1,\phi & A_2\acom\Gamma_2,\phi\Ra\Delta_2
  }
  \qquad 
  \infer[{[\text{Cut\/}]}]{A_1,A_2\acom\Gamma_1,\Gamma_2\Ra\Delta_1,\Delta_2}{
    A_1\acom\Gamma_1\Ra\Delta_1,\phi & A_2,\phi\acom\Gamma_2\Ra\Delta_2
  }
\end{gather*}

Let $a = A\acom\Gamma\Ra\Delta$ be an argument. We continue using $\overline{a}$ and $A\acom\Gamma\not\Ra\Delta$ to denote that $a$ has been eliminated. Arguments are attacked in the set of assumptions, we give an example in the next section. Although many details are still missing, it is already possible to define assumptive sequent-based argumentation frameworks.

\begin{definition}
An \emph{assumptive sequent-based argumentation framework} for a set of formulas $\calS$, set of assumptions $\AS$, based on a logic \logic\ and a set $\ARule$ of sequent elimination rules, is a pair ${\sf AF}_{\sfL,\ARule}^\AS = \tuple{\ArgLSAS,\ATtack}$, where $\ATtack\subseteq\ArgLSAS\times\ArgLSAS$ and $(a_1,a_2)\in\ATtack$ iff there is an $\calR\in\ARule$ such that $a_1$ $\calR$-attacks $a_2$. 
\label{def:assseqAF}
\end{definition}

Like before, when these are clear from the context or arbitrary, we will omit the subscripts $\sfL$, $\ARule$ and/or $\AS$. 
The semantics, as defined in Definition~\ref{def:extension} can be applied to assumptive sequent-based argumentation frameworks. The corresponding entailment relations (from Definition~\ref{def:entailment}) are denoted by $\ncass{\sem}^\pi$ for $\pi\in\{\cap,\cup,\Cap\}$. 

%
%
%
%
%
%
\subsection{Maximally consistent subsets with assumptions}

To reflect the different premise sets in an assumptive framework $\AFLSAS$, we define $\MCS_\sfL(\calS,\AS)$.  Then $\calT\in\MCS_\sfL(\calS,\AS)$ iff $\calT\subseteq\AS$ and there is no $\calT\subset\calT'\subseteq\AS$ such that $\calT'\cup\calS$ is consistent. Thus, $\MCS_\sfL(\calS,\AS)$ is the set of all maximally consistent subsets of $\AS$ that are consistent with $\calS$. The entailment relations are adjusted as follows:

\begin{definition}
  \label{def:MCSassentailment}
  Let \logic, $\calS$ a consistent set of $\calL$-formulas and $\AS$ a set of assumptions. 
  \begin{itemize}
    \item $\calS\:\ncmcs^{\cap, \AS}\:\phi$ iff $\phi\in\CN(\bigcap\MCS(\calS,\AS)\cup\calS)$;
    \item $\calS\:\ncmcs^{\cup,\AS}\:\phi$ iff $\phi\in\bigcup_{\calT\in\MCS(\calS,\AS)}\CN(\calS\cup\calT)$;
    \item $\calS\:\ncmcs^{\Cap,\AS}\:\phi$ iff $\phi\in\bigcap_{\calT\in\MCS(\calS,\AS)}\CN(\calS\cup\calT)$.
  \end{itemize}
\end{definition}


A well-known approach in argumentation theory, in which defeasible assumptions play an essential role, is assumption-based argumentation (ABA)~\cite{BDKT97,DKT09,Toni14}. In the next section we show how ABA can be implemented in the introduced general framework.

\section{Incorporating ABA}
\label{sec:ABA}

Assumption-based argumentation (ABA) was introduced in~\cite{BDKT97}, see also~\cite{DKT09,Toni14}. It takes as input a formal deductive system, a set of assumptions and a contrariness mapping for each assumption. There are only few requirements placed on each of these, keeping the framework abstract on the one hand, while the arguments have a formal structure and the attacks are based on the latter. First some of the most important definitions for the ABA-framework, from~\cite{BDKT97}:

\begin{definition}
  \label{def:deductivesystem}
  A \emph{deductive system} is a pair $\tuple{\calL,\calR}$, where $\calL$ is a formal language and $\calR$ is a set of rules of the form $\phi_1,\ldots,\phi_n\rightarrow\phi$, for $\phi_1,\ldots,\phi_n,\phi\in\calL$ and $n\geq 0$. 
\end{definition}

\begin{definition}
  \label{def:ABAdeduction}
  A \emph{deduction} from a theory $\Gamma$ is a sequence $\psi_1,\ldots,\psi_m$, where $m>0$, such that for all $i=1,\ldots, m$, $\psi_i\in\Gamma$, or there is a rule $\phi_1,\ldots\phi_n\rightarrow\psi_i\in\calR$ with $\phi_1,\ldots,\phi_n\in\{\psi_1,\ldots,\psi_{i-1}\}$. We denote by $\Gamma\vdash^\calR\psi_m$ a deduction from $\Gamma$ using rules in $\calR$. It is assumed that $\Gamma$ is $\subseteq$-minimal.
\end{definition}

\begin{example}
  \label{ex:ABAdeductionseq}
  An example of a deductive system is classical logic, where $\phi_1,\ldots, \phi_n\rightarrow\phi\in\calR_\CL$ if and only if $\phi_1,\ldots, \phi_n\vdash_\CL\phi$. Thus, we have that $\Gamma\vdash^\calR\psi$ if and only if $\Gamma\vdash_\CL\psi$ (modulo minimality). 
\end{example}

From this ABA argumentation frameworks can be defined:

\begin{definition}
  \label{def:ABAframework}
  An \emph{ABA-framework} is a tuple $\AFABA=\tuple{\calL, \calR, \calS, \calA, \overline{\cdot}}$ where:
  \begin{itemize}
    \item $\tuple{\calL,\calR}$ is a deductive system;
    \item $\calS\subseteq\calL$ a set of formulas, that satisfies non-triviality ($\calS\nvdash^\calR\phi$ for all $\phi$ that do not share an atom with any of the formulas in $\calS$);\footnote{In the remainder, if a set of formulas $\calS$ satisfies non-triviality, it is said that $\calS$ is non-trivializing.}
    \item $\calA\subseteq\calL$ a non-empty set of \emph{assumptions} for which $\calS\cap\calA = \emptyset$; and
    \item $\overline{\cdot}$ a mapping from $\calA$ into $\calL$, where $\overline{\phi}$ is said to be the \emph{contrary} of $\phi$. 
  \end{itemize}
\end{definition}

A simple way of defining contrariness in the context of classical logic is by $\overline{\phi} = \neg\phi$. 

\begin{definition}
  \label{def:ABAconsistency}
  Given an ABA-framework $\AFABA$, a set $A\subseteq\calA$ is:
  \begin{itemize}
    \item \emph{consistent} iff there is no $\phi\in A$ such that $A',\Gamma\vdash^\calR\overline{\phi}$ for some $A'\subseteq A$ and some $\Gamma\subseteq\calS$;
    \item \emph{maximally consistent} iff there is no $A'$ such that $A\subset A'\subseteq\calA$ and $A'$ is consistent, then $A\in\MCS(\calS,\calA)$. 
  \end{itemize}  
%
%
  The \emph{closure} of $\calT\subseteq\calL$ is defined as $\CN(\calT) = \{\phi\mid \Gamma\vdash^\calR\phi\text{ for }\Gamma\subseteq\calT\}$. 
\end{definition}

ABA-arguments are defined in terms of deductions and an attack is on the assumptions of the attacked argument. As in~\cite{DKT09}, arguments are not required to be consistent.  

\begin{definition}
  \label{def:ABAargument}
  Let $\AFABA=\tuple{\calL, \calR, \calS, \calA, \overline{\cdot}}$. An \emph{ABA-argument} for $\phi\in\calL$ is a deduction $A\cup\Gamma\vdash^\calR\phi$, where $A\subseteq\calA$ and $\Gamma\subseteq\calS$. The set $\ArgLRSA$ denotes the set of all ABA-arguments for $\calS$ and $\calA$.
\end{definition}

\begin{definition}
  \label{def:ABAattack}
  Let $\AFABA=\tuple{\calL, \calR, \calS, \calA, \overline{\cdot}}$. An argument $A\cup\calS\vdash^\calR\phi$ \emph{attacks} an argument $A'\cup\calS\vdash^\calR\phi'$ iff $\phi = \overline{\psi}$ for some $\psi\in A'$.
\end{definition}

The following requirement will be necessary for many of the proofs below. 

\begin{definition}
  \label{def:contrapass}
  $\vdash^\calR$ is \emph{contrapositive for assumptions}: for $\phi,\psi\in A$, $A\cup\Gamma\vdash^\calR\overline{\psi}$ if and only if $(A\setminus\{\phi\})\cup\{\psi\}\cup\Gamma\vdash^\calR\overline{\phi}$. 
\end{definition}

Semantics are defined as usual, see Definition~\ref{def:extension}. From this we can define the corresponding entailment relation:

\begin{definition}
  \label{def:ABAentailment}
  Let $\AFABA=\tuple{\calL, \calR, \calS, \calA, \overline{\cdot}}$ and $\sem\in\{\grd,\cmp,\prf,\stb\}$. 
  \begin{itemize}
    \item $\calA\cup\calS\:\ncaba{\sem}^\cup\:\phi$ ($\calA\cup\calS\:\ncaba{\sem}^\cap\:\phi$) if and only if for some (every) extension $\ext\in\exts_\sem(\AFABA)$ there is an argument $A\cup\Gamma\vdash^\calR\phi\in\ext$ for $A\subseteq\calA$ and $\Gamma\subseteq\calS$. 
    \item $\calA\cup\calS\:\ncaba{\sem}^\Cap\:\phi$ if and only if for every $\ext\in\exts_\sem(\AFABA)$ there is an $a\in\ext$ and $\conc(a) = \phi$. 
  \end{itemize}
\end{definition}

\begin{example}
  \label{ex:ABAframeworkdeduction}
  Recall the deductive system $\calR_\CL$ for classical logic, described in Example~\ref{ex:ABAdeductionseq} and let $\overline{\phi} = \{\neg\phi\}$. Consider the sets $\calS = \{s\}$ and $\calA = \{p, q, \neg p\vee\neg q, \neg p\vee r, \neg q\vee r\}$. Some of the arguments of $\AFABA$ are:\footnote{To avoid clutter, sometimes the superscript $\calR$ in $\vdash^\calR$ is omitted.} $a=s\vdash s$, $\ b=p, \neg p\vee\neg q\vdash\neg q$, $\ c = q, \neg p\vee\neg q\vdash\neg p$ and $\ d=p, q, \neg p\vee r, \neg q\vee r\vdash r.$  
 
  Note that $a$ cannot be attacked, since the set of assumptions of $a$ is empty. For the other arguments, we have that $b$ attacks $c$ and $d$, and $c$ attack $b$ and $d$. It can be shown that $\calA\cup\calS \:\ncaba{\sem}^\pi\:s$, for $\pi\in\{\cap, \cup, \Cap\}$, $\sem\in\{\grd, \cmp,\prf, \stb\}$. Furthermore, $\calA\cup\calS\:\ncaba{\sem}^\cup\:\phi$, but $\calA\cup\calS\:\nncaba{\sem}^\cap\:\phi$ and $\calA\cup\calS\:\nncaba{\sem}^\Cap\:\phi$ for $\sem\in\{\cmp,\prf, \stb\}$ and $\phi\in\{p, q, \neg p\vee\neg q\}$.
\end{example}

Based on the above notions from assumption-based argumentation, a corresponding sequent-based ABA-framework can be defined: 

\begin{definition}
  \label{def:seqABAframework}
  Let $\AFABA = \tuple{\calL, \calR, \calS, \calA, \overline{\cdot}}$ be an ABA-framework as defined above. The corresponding \emph{sequent-based ABA-framework} is then $\calAF_{\tuple{\calL,\calR_\Ra}}^{\sf ABA_\Ra}(\calS,\calA) =\tuple{\ArgLRseq(\calS,\calA),\ATtack}$, where:
  \begin{itemize}
    \item $\calR_\Ra$ is defined as:
    \begin{itemize}
      \item if $\tuple{\calL,\calR}$ is a logic with corresponding sequent calculus \calc, $\calR_\Ra = \mathsf{C}\cup\{\text{AS}_{\sf ABA}\}$ such that:
      \begin{small}
      \[
        \qquad\infer[\text{AS}_{\sf ABA}]{A,\phi\acom\Gamma\Ra\psi}{
          A\acom\Gamma,\phi\Ra\psi
        }
        \qquad
        \infer[\text{AS}_{\sf ABA}]{A\acom\Gamma,\phi\Ra\psi}{
          A,\phi\acom\Gamma\Ra\psi
        }
        \qquad
        \text{ where }\phi\in\calA.
      \]
      \end{small}
      \item otherwise $\calR_\Ra = \{\mu(r)\mid r\in\calR\}\cup\{\text{AS}_{\sf ABA},[\text{Cut\/}],[{\sf id}]\}$ where, for each $r = \phi_1,\ldots,\allowbreak \phi_n\rightarrow\phi\in\calR$, $\mu(r) = \phi_1,\ldots, \phi_n\Ra\phi$ and:     
      $
        \frac{}{\phi\Ra\phi}[\textsf{id}] 
      $ 
    \end{itemize}
    
    \item $a = A\acom\Gamma\Ra\phi\in\ArgLRseqSA$ for $A\subseteq\calA$, $\Gamma\subseteq\calS$ iff there is a derivation of $a$ using rules in $\calR_\Ra$.
    
    \item $(a_1, a_2)\in\ATtack$ iff $a_1$ $\calR$-attacks $a_2$ as defined in Definition~\ref{def:seqAF}, for $\ARule = \{\text{AT}_{\sf ABA}\}$ and:
    \begin{equation}
      \infer[\text{AT}_{\sf ABA}]{A_2,\overline{\phi}\acom\Gamma_2\not\Ra\psi}{
        A_1\acom\Gamma_1\Ra\overline{\phi} & A_2,\phi\acom\Gamma_2\Ra\psi
      }
    \end{equation}
  \end{itemize}
\end{definition}

\begin{remark}
  \label{rem:acomseq}
  $A\cup\Gamma\Ra\phi$ is derivable iff $A\acom\Gamma\Ra\phi$ is derivable. 
%
\end{remark}

%

In the next example we show how classical logic, with corresponding sequent calculus \LK\ can be taken as underlying deductive system. 

\begin{example}
  \label{ex:ABAseqCL}
  Let $\CL = \tuple{\calL,\vdash}$, where $\overline{\phi} = \neg \phi$ and $\calR_\Ra = \LK$. According to Definition~\ref{def:assargument} $A\acom\Gamma\Ra\phi\in\ArgL(\calS,\calA)$ iff $\Gamma\cup A\Ra\phi$ is derivable in \LK, for some finite $A\subseteq\calA$ and $\Gamma\subseteq\calS$. Since $\calR_\Ra = \LK\cup\{\text{AS}_{\sf ABA}\}$ it follows immediately that $A\cup\Gamma\Ra\phi$ is derivable in $\calR_\Ra$ iff it is derivable in \LK. 
  
%
\end{example}

In what follows let $\tuple{\calL,\calR}$ be a deductive system, $\calS\subseteq\calL$ a non-trivializing set of formulas and $\calA\subseteq\calL$ a set of assumptions, such that $\Gamma\subseteq\calS$ and $A\subseteq\calA$ are finite. Let $\AFseqABA =\tuple{\ArgLRseq(\calS,\calA),\ATtack}$ be a sequent-based ABA-framework and $\AFABA = \tuple{\calL, \calR, \calS, \calA,\overline{\cdot}}$. 

\begin{proposition}
  \label{prop:ABAsem}
  $\calA\cup\calS\:\ncaba{\sem}^\pi\:\phi$ iff $\calA\cup\calS\:\ncass{\sem}^\pi\:\phi$ for $\sem\in\{\grd, \cmp,\prf,\stb\}$ and $\pi\in\{\cup,\cap,\Cap\}$. 
\end{proposition}

The above proposition is a corollary of the following two lemmas:

\begin{lemma}
  \label{lem:ABAvsSeqargu}
  $A\cup\Gamma\vdash^\calR\phi\in\ArgLRSA$ iff $A\acom\Gamma\Ra\phi\in\ArgLRseq(\calS,\calA)$.
\end{lemma}

\begin{proof}
  Consider both directions:
  \begin{itemize}
    \item[$\Ra$] 
Assume that $A\cup\Gamma\vdash^\calR\phi\in\ArgLRSA$. Then there is a deduction from the theory $A\cup\Gamma$ for the formula $\phi$. By Definition~\ref{def:ABAdeduction}, there is a sequence $\psi_1,\ldots,\psi_m$ ($\psi_m = \phi$), such that for each $i=1,\ldots,m$, $\psi_i\in A\cup\Gamma$ or there is a rule $\phi_1,\ldots,\phi_n\rightarrow\psi_i=r\in R$ and $\phi_1,\ldots,\phi_n\in\{\psi_1,\ldots,\psi_{i-1}\}$. 
We proceed by induction on $m$, showing that for each $\psi_i$, there is a sequent $s_i = A_i\cup\Gamma_i\Ra\psi_i$:
\begin{itemize}
  \item $m=1$. Then either $\psi_1\in A\cup\Gamma$ and thus $\psi_1\Ra\psi_1$ is derivable in $\calR_\Ra$, by $[{\sf id}]$. Or there is a rule $\ \rightarrow\psi_1\in R$. Hence $\ \Ra\psi_1\in\calR_\Ra$ for  $A \cup\Gamma = \emptyset$. Since $\psi_1 = \psi_m = \phi$, $A\cup\Gamma\Ra\phi$ is derivable. 
  \item $m=k+1$. Assume that for sequences up to $k\geq 1$, for each $\psi_i$ there is a sequent $s_i = A_i\cup\Gamma_i\Ra\psi_i$. Now consider $\psi_{k+1}$. Then $\psi_{k+1} \in A\cup\Gamma$, from which it follows immediately that $A\cup\Gamma\Ra\psi_{k+1}$ is derivable in $\calR_\Ra$, or there is a rule $\phi_1,\ldots, \phi_n\rightarrow\psi_{k+1} = r\in\calR$ and $\phi_1,\ldots,\phi_n\in\{\psi_1,\ldots,\psi_k\}$. By Definition~\ref{def:seqABAframework}, $\phi_1,\ldots,\phi_n\Ra\psi_{k+1}\in\calR_\Ra$. Furthermore, by induction hypothesis, for each $\psi_i\in\{\psi_1,\ldots,\psi_k\}$, there is a sequent $s_i = A_i\cup\Gamma_i\Ra\psi_i$.  Hence, $\phi_1,\ldots,\phi_n\in\{\conc(s_1),\ldots, \conc(s_k)\}$. By applying $[\text{Cut\/}]$ we obtain a sequent $s_{k+1} = A_{k+1}\cup\Gamma_{k+1}\Ra\psi_{k+1}$. 
\end{itemize}
      
    Hence, there is a sequence of sequents $s_1,\ldots,s_m$, such that $s_i$ is derived from $s_1,\ldots, s_{i-1}$ by applying rules in $\calR_\Ra$ and $s_m = A\cup\Gamma\Ra\phi$. That $A\acom\Gamma\Ra\phi\in\ArgLRseq(\calS,\calA)$ follows by Remark~\ref{rem:acomseq}. 
    
    \item[$\Leftarrow$] Now suppose that $a=A\acom\Gamma\Ra\phi\in\ArgLRseq(\calS,\calA)$. By Remark~\ref{rem:acomseq}, $A\cup\Gamma\Ra\phi$ is derivable in $\calR_\Ra$ as well. Then there is a derivation via a sequence of sequents $s_1,\ldots, s_m$, where $s_i = A_i\cup\Gamma_i\Ra\psi_i$ for each $i\in\{1,\ldots,m\}$ is the result of applying rules from $\calR_\Ra$ to sequents in $\{s_1,\ldots, s_{i-1}\}$ and $s_m = A\cup\Gamma\Ra\phi$. Again by induction on the length of the derivation $m$, for each $s_i$, there is a deduction $\ass(s_i) \cup\supp(s_i)\vdash^\calR\conc(s_i)$ via the sequence $\Phi_i = \psi^i_1,\ldots,\psi^i_{m_i}$:
    \begin{itemize}
      \item $m=1$. Then $\phi\in A\cup\Gamma$ in which case $s_m = \phi\Ra\phi$ or there is a $\mu(r)\in\calR_\Ra$ such that $\mu(r)= \ \Ra\phi$ and thus, by Definition~\ref{def:seqABAframework}, $r =\ \rightarrow\phi\in\calR$. Hence $A\cup\Gamma\vdash^\calR\phi$.      
      \item $m=k+1$. Now assume that for derivations up to length $k\geq 1$, for each $s_i$, there is a deduction from $\ass(s_i)\cup\supp(s_i)$ for $\conc(s_i)$ via the sequence $\Phi_i$. That $s_m$ is derivable implies that $\conc(s_m)\in\ass(s_m)\cup\supp(s_m)$, in which case $s_m = \conc(s_m)\Ra\conc(s_m)$, from which it follows immediately that there is a deduction $\ass(s_m)\cup\supp(s_m)\vdash^\calR\conc(s_m)$ or $s_m$ is the result of applying a rule to sequents in $\{s_1,\ldots, s_k\}$:
      \begin{itemize}
        \item suppose that $[\text{Cut\/}]$ was applied to $s_{j_1},s_{j_2}\in\{s_1,\ldots, s_k\}$. By induction hypothesis, there are deductions $\ass(s_{j_1})\cup\supp(s_{j_1})\vdash^\calR\conc(s_{j_1})$ and $\ass(s_{j_2})\cup\supp(s_{j_2})\vdash^\calR\conc(s_{j_2})$ via the sequence $\Phi_{j_1}$ respectively $\Phi_{j_2}$. The deduction $\ass(s_m)\cup\supp(s_m)\vdash^\calR\conc(s_m)$ is obtained via the sequence $\Phi_m = \Phi_{j_1}\circ_{\conc(s_{j_1})}\Phi_{j_2}$, where $\Phi^1\circ_\psi\Phi^2$ denotes the concatenation of $\Phi^1$ with $\Phi^2$ such that all occurrences of $\psi$ in $\Phi^2$ are taken out. 
        \item suppose that $s_m$ is the result of applying $\phi_1,\ldots, \phi_n\Ra\phi = \mu(r)\in\calR_\Ra$. By construction, $\phi_1,\ldots, \phi_n\rightarrow\phi = r\in\calR$ such that $\phi_j\in\{\psi_1,\ldots, \psi_k\}$ is obtained via a sequence $\Phi'_j$, for each $j\in\{1,\ldots, n\}$. Therefore, $\ass(s_m)\cup\supp(s_m)\vdash^\calR\conc(s_m)$.
      \end{itemize}
    \end{itemize}
    Thus, for the derivation of $a$, of any length $m$, via the sequence of sequents, $s_1,\ldots, s_m$, there is a deduction from $A\cup\Gamma$ via the sequence $\Phi_m$, for $\phi$. Hence $A\cup\Gamma\vdash^\calR\phi\in\ArgLR(\calS,\calA)$. 
  \end{itemize}
\end{proof}


\begin{lemma}
  \label{lem:ABAvsSeqattack}
  Let $a,b\in\ArgLRSA$ and $a',b'$ their corresponding ABA-sequent arguments, thus $a',b'\in\ArgLRseqSA$.\footnote{That $a'$ and $b'$ exist follows from Lemma~\ref{lem:ABAvsSeqargu}.} Then $a$ attacks $b$ in $\AFABA$ iff $a'$ attacks $b'$ in $\AFseqABA$.
\end{lemma}

\begin{proof}
  Consider the $\Ra$-direction, the $\Leftarrow$-direction is similar and left to the reader.

  Let $a,b\in\ArgLRSA$ and assume $a = A\cup\Gamma\vdash^\calR\phi$ attacks $b=A'\cup\Gamma'\vdash^\calR\phi'$. Then, by Definition~\ref{def:ABAattack}, $\phi = \overline{\psi}$ for $\psi\in A'$. By Lemma~\ref{lem:ABAvsSeqargu}, $a' = A\acom\Gamma\Ra\phi$ and $b' = A' \acom\Gamma'\Ra\phi'$ are arguments in $\AFseqABA$ ($a',b'\in\ArgLRseq(\calS,\calA)$). Since $\phi=\overline{\psi}$ for $\psi\in A'$, it follows that $a'$ $\text{AT}_{\sf ABA}$-attacks $b'$.
\end{proof}

\begin{example}
  \label{ex:ABAprf}
  Recall the setting from Example~\ref{ex:ABAframeworkdeduction}, in which $\calS = \{s\}$, $\calA  = \{p, q, \neg p\vee\neg q, \neg p\vee r, \neg q\vee r\}$ and classical logic the core logic. Let $\calR_\Ra = \LK$, for $\AFseqABA =\tuple{\ArgLRseqSA,\ATtack}$, some of the arguments in $\ArgLRseqSA$ are: $\ a = s\Ra s$, $\ b= p, \neg p\vee\neg q\acom\ \Ra\neg q$, $\ c = q, \neg p\vee\neg q\acom\ \Ra\neg p$ and $\ d = p,q, \neg p\vee r, \neg q\vee r\acom\ \Ra r.$
  
  Note that $a$ cannot be attacked, since $\ass(a)=\emptyset$. We thus have $\calA\cup\calS\:\ncass{\sem}^\pi\: s$ for $\sem\in\{\grd,\cmp,\prf,\allowbreak\stb\}$ and $\pi\in\{\cup,\cap,\Cap\}$. However, the argument $d$ is attacked by both $b$ and $c$. Moreover $b$ attacks $c$ and $c$ attacks $b$. It can be shown that, for $\phi\in\{p, q, \neg p\vee\neg q\}$, $\calA\cup\calS\:\nncass{\sem}^\pi\:\phi$ for $\sem\in\{\grd,\cmp,\prf,\allowbreak\stb\}$ and $\pi\in\{\cap,\Cap\}$ but also $\calA\cup\calS\:\ncass{\sem'}^\cup\:\phi$ for $\sem'\in\{\cmp,\prf,\stb\}$. 
\end{example}

The relations between ABA and reasoning with maximally consistent subsets and between sequent-based argumentation and maximally consistent subsets have been studied~\cite{ABS17IEAAIE,ABS17AI,HeyAri18}. In addition to the two entailment relations in~\cite{HeyAri18} (in the notation of this paper $\ncmcs^{\Cap,\calA}$ and $\ncmcs^{\cup,\calA}$), the entailment relation $\ncmcs^{\cap,\calA}$ is considered below as well. Moreover, the semantics as defined in this paper is based on sets of arguments, were as in~\cite{HeyAri18}, sets of assumptions make up the extensions.  
The proof of Proposition~\ref{prop:ABAmcs}, and the lemmas necessary for it, are based on proofs in~\cite{ABS17IEAAIE,ArStr16KR}. 

\begin{proposition}
  \label{prop:ABAmcs}
  Let $\AFseqABA$ for a deductive system $\tuple{\calL,\calR}$, $\calS\subseteq\calL$ a non-trivializing set of formulas and $\calA$ a set of assumptions. Then: $\calA\cup\calS\ncass{\prf}^\pi\:\phi$ iff $\calA\cup\calS\ncass{\stb}^\pi\:\phi$ iff $\calS\:\ncmcs^{\pi, \calA}\:\phi$, for $\pi\in\{\cap,\cup,\Cap\}$.
\end{proposition}


\begin{lemma}
  \label{lem:consistencyMCS}
  For each set $\calT\subseteq\calA$: $\calT\in\MCS(\calS,\calA)$ iff for each $\phi\in\calA\setminus\calT$, there is some finite $A\subseteq\calT$ and some finite $\Gamma\subseteq\calS$ such that $A\acom\Gamma\Ra\overline{\phi}\in\ArgLRseqSA$. 
\end{lemma}

\begin{proof}
($\Ra$) Assume that $\calT\in\MCS(\calS,\calA)$ and consider some $\phi\in\calA\setminus\calT$. By Definition~\ref{def:ABAconsistency}, there is some $A'\subseteq\calT\cup\{\phi\}$ and some $\Gamma\subseteq\calS$ such that $A'\cup\Gamma\vdash^\calR\overline{\psi}$ for some $\psi\in\calT\cup\{\phi\}$. Consider two cases: (a) $\psi\in\calT$, then by contraposition, $(A'\setminus\{\phi\})\cup\{\psi\}\cup\Gamma\vdash^\calR\overline{\phi}$; and (b) $\psi = \phi$. Then $A'\subseteq\calT$. 
    
    In both cases there is an $A\subseteq\calT$ and a $\Gamma\subseteq\calS$ such that $A\cup\Gamma\vdash^\calR\overline{\phi}\in\ArgLRSA$. Hence, by Lemma~\ref{lem:ABAvsSeqargu}, $A\acom\Gamma\Ra\overline{\psi}\in\ArgLRseqSA$.     

($\Leftarrow$) Now assume that for each $\phi\in\calA\setminus\calT$, there is some finite $A\subseteq\calT$ and some finite $\Gamma\subseteq\calS$ such that $A\acom\Gamma\Ra\overline{\phi}\in\ArgLRseqSA$, Hence, by Lemma~\ref{lem:ABAvsSeqargu}, $A\cup\Gamma\vdash^\calR\overline{\phi}\in\ArgLRSA$. It follows that for each $\phi\in\calA\setminus\calT$, there are $A\subseteq\calT\cup\{\phi\}$ and $\Gamma\subseteq\calS$ such that $A\cup\Gamma\vdash^\calR\overline{\psi}$ for $\psi\in\calT\cup\{\phi\}$. Hence $\calT$ is maximally consistent.  
\end{proof}

\begin{lemma}
  \label{lem:ABAconsistentass}
  The set $\ass(\ext)$, for any $\ext\in\exts_\cmp(\AFseqABA)$ is consistent. 
\end{lemma}

\begin{proof}
  Assume, towards a contradiction, that $\ass(\ext)=\{\phi_1,\ldots,\phi_n\}$ is not consistent. Then, by Definition~\ref{def:ABAconsistency} and Lemma~\ref{lem:ABAvsSeqargu}, $a=A\acom\quad\Ra\overline{\phi_i}$ is derivable for some $A\subseteq\ass(\ext)$ and $i\in\{1\ldots,n\}$. Suppose that $a$ is attacked by an argument $b = A'\acom\Gamma\Ra\psi\in\ArgLRseqSA$. Then $\psi = \overline{\psi'}$ for some $\psi'\in A$. Hence $\psi'\in\ass(\ext)$. Thus $b$ attacks some argument $a'\in\ext$ as well. Since $a'\in\ext$, there is an argument $c\in\ext$ which defends $a'$ and thus $a$ from the attack by $b$. Since $\ext$ is complete, $a\in\ext$. However, $a$ attacks each $a_j\in\ext$ with $\phi_i\in\ass(a_j)$. A contradiction with the conflict-freeness of the complete extension $\ext$. 
\end{proof}


\begin{lemma}
  \label{lem:ABAMCSStb}
  If $\calT\in\MCS(\calS,\calA)$ then $\ArgLRseq(\calS,\calT)\in\exts_\stb(\AFseqABA)$. 
\end{lemma}

\begin{proof}
  Assume that $\calT\in\MCS(\calS,\calA)$ and let $\ext =\ArgLRseq(\calS,\calT)$. Suppose $\ext$ is not conflict-free. Then there are arguments $a_1 = A_1\acom\Gamma_1\Ra\phi_1$ and $a_2 = A_2\acom\Gamma_2\Ra\phi_2$, such that $a_1,a_2\in\ext$ and $a_1$ attacks $a_2$. Thus $\phi_1=\overline{\psi}$ for some $\psi\in A_2$. However, by assumption $A_1\cup A_2\subseteq\calT$. A contradiction with the assumption that $\calT\in\MCS(\calS,\calA)$.
  
  Now suppose that $A'\acom\Gamma'\Ra\phi'\in\ArgLRseq(\calS,\calA)\setminus\ext$ for some $\Gamma'\subseteq\calS$ and $A'\subseteq\calA$. Thus there is some $\phi\in A'\setminus \calT$. Since, by supposition $\calT\in\MCS(\calS,\calA)$, from Lemma~\ref{lem:consistencyMCS}, there are finite $A\subseteq\calT$, $\Gamma\subseteq\calS$ such that $A\acom\Gamma\Ra\overline{\phi}\in\ArgLRseqSA$. Because $A\subseteq\calT$, $A\acom\Gamma\Ra\overline{\phi}\in\ext$. Hence $A'\acom\Gamma'\Ra\phi'$ is attacked by $\ext$. Therefore $\ext$ attacks every argument in $\ArgLRseq(\calS,\calA)\setminus\ext$ and thus $\ext\in\exts_\stb(\AFseqABA)$. 
\end{proof}

\begin{lemma}
  \label{lem:ABAPrfMCS}
  If $\ext\in\exts_\prf(\AFseqABA)$ then there is some $\calT\in\MCS(\calS,\calA)$ such that $\ext = \ArgLRseq(\calS,\calT)$. 
\end{lemma}

\begin{proof}
  Suppose, towards a contradiction, that for some extension $\ext\in\exts_\prf(\AFseqABA)$ there is no $\calT\in\MCS(\calS,\calA)$ such that $\ext = \ArgLRseq(\calS,\calT)$. Then there is no $\calT\in\MCS(\calS,\calA)$ such that $\ass(\ext)\subseteq\calT$ and hence, $\ass(\ext)$ is inconsistent. A contradiction with Lemma~\ref{lem:ABAconsistentass} and the supposition that $\ext$ is a preferred extension. Thus, $\ext\subseteq\ArgLRseq(\calS,\calT)$ for some $\calT\in\MCS(\calS,\calA)$. By Lemma~\ref{lem:ABAMCSStb}, $\ArgLRseq(\calS,\calT)$ is stable and thus $\ext = \ArgLRseq(\calS,\calT)$.
\end{proof}

We now turn to the proof of Proposition~\ref{prop:ABAmcs}:

\begin{proof}
  Let $\AFseqABA$ for $\tuple{\calL,\calR}$ a deductive system, $\calS$ a non-trivializing set of $\calL$-formulas, $\calA$ a set of assumptions. 
  Consider each item in both directions:

  \begin{enumerate}
    \item ($\Ra$) Note that $\calA\cup\calS\:\ncass{\prf}^\cap\:\phi$ implies $\calA\cup\calS\:\ncass{\stb}^\cap\:\phi$. Suppose $\calS\:\nncmcs^{\cap,\calA}\:\phi$, but that there is some finite $A\subseteq\calA$ and some $\Gamma\subseteq\calS$ such that $A\acom\Gamma\Ra\phi\in\ArgLRseqSA$. Now, by assumption, $A\not\subseteq\bigcap\MCS(\calS, \calA)$. Hence, there is some $\phi'\in A\setminus \bigcap\MCS(\calS, \calA)$. From which it follows that there is some $\calT\in\MCS(\calS,\calA)$ such that $\phi'\notin\calT$. Therefore $A\acom\Gamma\Ra\phi\notin\ArgLRseq(\calS,\calT)$. By Lemma~\ref{lem:ABAMCSStb}, $\ArgLRseq(\calS,\calT)\in\exts_\stb(\AFseqABA)$, thus $\calA\cup\calS\:\nncaba{\stb}^\cap\:\phi$ (and thus $\calA\cup\calS\:\nncaba{\prf}^\cap\:\phi$) as well. 
    
    ($\Leftarrow$) Suppose that $\calS\:\ncmcs^{\cap,\calA}\:\phi$. Thus, there are finite $A\subseteq\bigcap\MCS(\calS,\calA)$ and $\Gamma\subseteq\calS$ such that $A\acom\Gamma\Ra\phi\in\ArgLRseq(\calS,\calA)$ is derivable. By Lemma~\ref{lem:ABAPrfMCS} $\ArgLRseq(\Gamma, A)\subseteq\bigcap\exts_\prf(\AFseqABA)$. Hence $A\acom\Gamma\Ra\phi\in\bigcap\exts_\prf(\AFseqABA)$. From which it follows that $\calA\cup\calS\:\ncass{\prf}^\cap\:\phi$ and thus $\calA\cup\calS\:\ncass{\stb}^\cap\:\phi$.   
    \item ($\Ra$) Note that $\calA\cup\calS\:\ncaba{\stb}^\cup\:\phi$ implies $\calA\cup\calS\:\ncaba{\prf}^\cup\:\phi$. Suppose that $\calA\cup\calS\:\ncaba{\prf}^\cup\:\phi$. Then there is some $\ext\in\exts_\prf(\AFseqABA)$ such that $A\acom\Gamma\Ra\phi\in\ext$, for $A\subseteq \calA$ and $\Gamma\subseteq\calS$. From Lemma~\ref{lem:ABAPrfMCS} it follows that there is some $\calT\in\MCS(\calS,\calA)$ such that $\ext = \ArgLRseq(\calS,\calT)$ (thus $A\subseteq\calT$). Hence, by Definition~\ref{def:ABAconsistency} and Lemma~\ref{lem:ABAvsSeqargu}, $\phi\in\CN(\calT\cup\calS)$ it follows that $\calS\:\ncmcs^{\cup,\calA}\:\phi$. 
    
    ($\Leftarrow$) Assume that $\calS\:\ncmcs^{\cup,\calA}\:\phi$. Then there is some $\calT\in\MCS(\calS,\calA)$ such that $\phi\in\CN(\calT\cup\calS)$. Therefore, there is a deduction from $A\cup\Gamma\subseteq\calT\cup\calS$ for $\phi$ ($A\cup\Gamma\vdash^\calR\phi\in\ArgLRSA$) and thus, by Lemma~\ref{lem:ABAvsSeqargu} $A\acom\Gamma\Ra\phi\in\ArgLRseqSA$. From Lemma~\ref{lem:ABAMCSStb} it follows that $\ArgLRseq(\calS,\calT)\in\exts_\stb(\AFseqABA)$. Thus $\calA\cup\calS\:\ncaba{\stb}^\cup\:\phi$ as well. 
    \item $\calA\cup\calS\:\ncaba{\stb}^\Cap\:\phi$ implies $\calS\:\ncmcs^{\Cap,\calA}\:\phi$: suppose that $\calS\:\nncmcs^{\Cap,\calA}\:\phi$, then there is some $\calT\in\MCS(\calS,\calA)$ for which $\phi\notin\CN(\calS\cup\calT)$. Hence, there are no $A\subseteq\calT$ and $\Gamma\subseteq\calS$ with $A\acom\Gamma\Ra\phi\in\ArgLRseq(\calS,\calT)$. From Lemma~\ref{lem:ABAMCSStb} it follows that $\ArgLRseq(\calS,\calT)\in\exts_\stb(\AFseqABA)$, thus $\calA\cup\calS\:\nncaba{\stb}^\Cap\:\phi$. 
    
    $\calS\:\ncmcs^{\Cap,\calA}\:\phi$ implies $\calA\cup\calS\:\ncaba{\prf}^\Cap\:\phi$: suppose that $\calA\cup\calS\:\nncaba{\prf}^\Cap\:\phi$. Then there is some extension $\ext\in\exts_\prf(\AFseqABA)$ such that there is no $A\acom\Gamma\Ra\phi\in\ext$ for $A\subseteq \calA$ and $\Gamma\subseteq\calS$. From Lemma~\ref{lem:ABAPrfMCS} it follows that there is some $\calT\in\MCS(\calS,\calA)$ such that $\ArgLRseq(\calS,\calT) = \ext$ and $\phi\notin\CN(\calS\cup\calT)$. Thus $\calS\:\nncmcs^{\Cap,\calA}\:\phi$. 
    
    $\calA\cup\calS\:\ncaba{\prf}^\Cap\:\phi$ implies $\calA\cup\calS\:\ncaba{\stb}^\Cap\:\phi$: this follows immediately since any stable extension is a preferred extension~\cite[Lemma~15]{Dung95}. 
  \end{enumerate}
\end{proof}

\begin{example}
  \label{ex:ABAMCS}
  Recall from Example~\ref{ex:ABAframeworkdeduction} the sets $\calS = \{s\}$ and $\calA = \{p, q, \neg p\vee\neg q, \neg p\vee r, \neg q\vee r\}$. Then $\MCS(\calS,\calA) = \{\{p, q, \neg p\vee r, \neg q\vee r\}, \{p, \neg p\vee\neg q, \neg p\vee r, \neg q\vee r\}, \{q, \neg p\vee\neg q, \neg p\vee r, \neg q\vee r\}\}$. Hence $\bigcap\MCS(\calS,\calA) = \{\neg p\vee r, \neg q\vee r\}$. Therefore, $\calS\:\ncmcs^{\pi, \calA}\:\phi$ for $\pi\in\{\cap,\Cap\}$ and $\phi\in\CN(\{s, \neg p\vee r, \neg q\vee r\})$ and $\calS\:\ncmcs^{\cup,\calA}\:\phi$ for $\phi\in\CN(\calS\cup\calA)$. 
\end{example}

\section{Conclusion}
\label{sec:Conclusion}

In order to allow for reasoning with assumptions, sequent-based argumentation was extended by adding a component for assumptions to each argument, resulting in assumptive sequent-based argumentation. As in sequent-based argumentation, any logic, with a corresponding sound and complete sequent calculus, can be taken as the core logic. 
Due to its generic and modular setting, assumptive sequent-based argumentation is more general than other approaches to reasoning with assumptions, such as assumption-based argumentation (where arguments are constructed by applying modus ponens to an inferential database and for which it was shown that it can be embedded in the here introduced framework), default assumptions~\cite{Mak03} (defined in terms of classical logic) and adaptive logics~\cite{Bat07,Str14} (based on a supra-classical Tarskian logic). Moreover, the proofs in the paper do not rely on the concrete nature of the underlying core logic. It therefore paves the way to equip many well-known logics (e.g., intuitionist logic and many modal logics) with defeasible assumptions.


From here, many future research directions can be taken. For example, the availability of first-order sequent calculi opens up a line of research into first-order generalizations and thus into nonmonotonic systems such as circumscription. Moreover, preferences among assumptions will be investigated. Recently, the relation between different nonmonotonic reasoning systems have been studied, see for an overview~\cite{HeyStr16}. There translations from ASPIC$^+$~\cite{Pra10} and adaptive logics into ABA are provided as well. Though it remains an open question to see how sequent-based argumentation fits within this group of nonmonotonic reasoning systems, these translations suggest that assumptive sequent-based argumentation is expressive enough to capture ASPIC$^+$ and adaptive logics.

\bibliographystyle{plain}
\bibliography{sequent}

\end{document}